\documentclass[10pt]{article}

\usepackage[paper=a4paper,left=23mm,right=23mm,top=28mm, bottom=24mm]{geometry}

\usepackage{hyperref}

\usepackage{authblk}
\usepackage{tikz}
\usepackage{xfrac}
\usepackage{tabularx, multirow}
\usepackage{graphicx, tikz}
\usetikzlibrary{arrows.meta}
\usetikzlibrary{arrows,calc}
\usetikzlibrary{shapes}
\usepackage{latexsym,amssymb,amsfonts,graphicx,makeidx,paralist}
\usepackage{theorem,rotating,ifthen,amsmath,subfigure,epsfig,nicefrac}
\usepackage{booktabs,framed}

\usepackage{hyperref}

\usepackage{xcolor}
\colorlet{shadecolor}{gray!12}

\newcommand{\g} {\mbox{digraph}}

\newtheorem{remark}{Remark}
\newtheorem{proposition}{Proposition}

\newenvironment{desctight}
  {\begin{list}{}{\setlength\labelwidth{0pt}
        \setlength{\itemsep}{0.5pt}
        \setlength{\parsep}{0pt}
        \setlength\itemindent{-\leftmargin}
        }}
    {\end{list}}

\newtheorem{theorem}{Theorem}[section]
\newtheorem{example}[theorem]{Example}
\newtheorem{corollary}[theorem]{Corollary}
\newtheorem{lemma}[theorem]{Lemma}
\newtheorem{observation}[theorem]{Observation}

\newtheorem{definition}[theorem]{Definition}

\newenvironment{proof}{\noindent{\bf Proof~}}{\null\hfill $\Box$\par\medskip}

\newcommand{\bigo}{\text{$\mathcal O$}}

\newcommand{\MSOA}{\text{MSO}_1}

\newcommand{\ideg}{\text{indegree}}
\newcommand{\odeg}{\text{outdegree}}
\newcommand{\un} {{\it un}}

\newcommand{\LMSOA}{\text{LinEMSO}_1}

\newcommand{\dcws} {\text{d-cw}}

\newcommand{\fpt} {\mbox{FPT}}
\newcommand{\xp} {\mbox{XP}}
\newcommand{\w} {\mbox{W}}

\newcommand{\DCN} {\text{DCN}}

\newcommand{\reach} {\text{reach}}
\newcommand{\sym} {\mbox{sym}}
\newcommand{\asym} {\mbox{asym}}
\newcommand{\dmws} {\text{dmw}}
\newcommand{\arc} {\mbox{arc}}
\newcommand{\p} {\mbox{P}}
\newcommand{\np} {\mbox{NP}}

\begin{document}


\title{Acyclic coloring of special digraphs\thanks{An extended abstract of this paper will appear in the
Proceedings of the 7th Annual International Conference on Algorithms and Discrete Applied Mathematics (CALDAM 2021) \cite{GKR21}.}}

\author{Frank Gurski}
\author{Dominique Komander}
\author{Carolin Rehs}

\affil{\small University of  D\"usseldorf,
Institute of Computer Science, Algorithmics for Hard Problems Group,\newline 
40225 D\"usseldorf, Germany}

\maketitle


\begin{abstract}
An acyclic $r$-coloring of a directed graph $G=(V,E)$ is a
partition of the vertex set $V$ into $r$ acyclic sets. The dichromatic
number of a directed graph $G$ is the  smallest $r$ such that $G$
allows an acyclic $r$-coloring. For symmetric digraphs the dichromatic number equals
the well-known chromatic number of the underlying undirected graph.
This allows us to carry over the $\w[1]$-hardness and lower bounds for running times
of the chromatic number problem parameterized by clique-width
to the dichromatic number problem parameterized by directed
clique-width.
We introduce the first polynomial-time algorithm for the acyclic coloring problem
on digraphs of constant directed clique-width. 
From a parameterized point of view our algorithm shows that
the  Dichromatic Number problem  is in $\xp$ when parameterized by directed clique-width 
and extends the only known structural parameterization by directed modular width for this problem.
Furthermore, we apply defineability within monadic second order logic 
in order to show that Dichromatic Number problem is in $\fpt$ when parameterized by the directed clique-width 
and $r$.

For directed co-graphs, which is a class of digraphs of directed clique-width 2, and several
generalizations we even show linear time solutions for computing the
dichromatic number. Furthermore, we conclude that directed co-graphs and
the considered generalizations lead to subclasses of perfect digraphs.
For  directed cactus forests, which is a set of digraphs of directed tree-width 1, 
we conclude  an upper bound of $2$ for the
dichromatic number and we show that  an optimal acyclic coloring 
can be computed in linear time.

\end{abstract}


\section{Introduction}

Within undirected graphs $G=(V,E)$ a {\em $r$-coloring} 
corresponds to a partition of the vertex set $V$ into $r$ independent sets. 
The smallest $r$ such that a graph
$G$ has a $r$-coloring is denoted as the {\em chromatic number} of $G$.
Even deciding whether a graph has a $3$-coloring is NP-complete but 
there are many efficient solutions for the coloring
problem on special graph classes. Among these are  chordal graphs \cite{Gol80},  
comparability graphs \cite{Hoa94}, and co-graphs \cite{CLS81}.

For oriented graphs the concept of oriented colorings, which has been introduced  by Courcelle \cite{Cou94}, received a lot of attention in \cite{CD06,Sop16,GH10,GHKLOR14,CFGK16,GKR19d,GKL20,GKL21}.
An {\em oriented $r$-coloring} of an oriented graph $G=(V,E)$ is a partition of the vertex set $V$ into $r$ independent sets, such
that all the arcs linking two of these subsets have the same direction.
The {\em oriented chromatic number} of an oriented graph $G$, denoted by $\chi_o(G)$, is the  smallest $r$ such that $G$
allows an oriented $r$-coloring.

In this paper, we consider an approach for coloring the vertices of directed graphs, introduced by Neumann-Lara \cite{NL82}.
An {\em acyclic $r$-coloring} of a digraph $G=(V,E)$ is a partition of the
vertex set $V$ into $r$  sets such that all sets  induce an acyclic subdigraph in $G$.
The {\em dichromatic number} of $G$ is the smallest integer $r$
such that $G$ has an acyclic $r$-coloring.
%
%
%
Acyclic colorings of digraphs received a lot of attention in \cite{BFJKM04,Moh03,NL82}
and also in recent works \cite{LM17,MSW19,SW20}. The dichromatic number is one of two basic concepts
for the class of perfect digraphs \cite{AH15} and can be regarded as a natural
counterpart of the well known chromatic number for undirected graphs.

In the Dichromatic Number problem ($\DCN$) there is given a digraph
$G$ and an integer $r$ and the question is whether $G$ has an acyclic $r$-coloring.
If $r$ is constant and not part of the input, the corresponding problem
is denoted by $\DCN_{r}$. Even  $\DCN_{2}$ is NP-complete  \cite{FHM03},
which motivates to consider the Dichromatic
Number problem on special graph classes. 
Up to now, only few classes of digraphs are known for which the
dichromatic number can be found in polynomial time.
The set of DAGs is obviously equal to the set of digraphs
of  dichromatic  number  $1$. Further,
every odd-cycle free digraph \cite{NL82} and every non-even digraph~\cite{MSW19}
has  dichromatic number at most $2$.

The hardness of the Dichromatic Number problem remains true, even  for inputs
of bounded  directed feedback vertex set size \cite{MSW19}.
This result implies that there are no $\xp$-algorithms\footnote{XP is the class
of all parameterized problems which can be solved by algorithms that are polynomial
if the parameter is considered as a constant \cite{DF13}.}  for the Dichromatic Number problem parameterized by directed width
parameters such as directed path-width, directed tree-width, DAG-width or Kelly-width, since all of these are upper bounded in terms of the size of a smallest feedback vertex set.
The first positive result concerning structural parameterizations of the Dichromatic Number problem
is the existence of an $\fpt$-algorithm\footnote{FPT is the class of all parameterized problems which can
be solved by algorithms that are exponential only in the size of a fixed parameter while
polynomial in the size of the input size \cite{DF13}.} for the  Dichromatic Number problem
parameterized by directed modular width \cite{SW19}.

In this paper, we introduce the first polynomial-time algorithm for the Dichromatic Number problem
on digraphs of constant directed clique-width.
Therefore, we consider a directed clique-width expression $X$
of the input digraph $G$ of directed clique-width $k$.
For each node $t$ of the corresponding rooted expression-tree $T$ we use  label-based reachability
information about the subgraph $G_t$ of the subtree rooted at $t$.
For every  partition of the vertex set of $G_t$ into acyclic
sets $V_1,\ldots,V_s$ we compute the multi set
$\langle \reach(V_1),\ldots,\reach(V_s) \rangle$, where $\reach(V_i)$, $1\leq i \leq s$,
is the set of all label pairs $(a,b)$ such that the subgraph of $G_t$ induced
by $V_i$ contains a vertex labeled by $b$, which is reachable by a vertex labeled by $a$.
By using bottom-up dynamic programming along expression-tree $T$,
we obtain an algorithm for the Dichromatic Number problem of running
time $n^{2^{\bigo(k^2)}}$ where $n$ denotes the number of vertices of the input
digraph. Since any algorithm with running time in $n^{2^{o(k)}}$  would disprove the Exponential Time Hypothesis (ETH),  the exponential dependence on $k$ in the degree
of the polynomial cannot be avoided, unless ETH fails.

From a parameterized point of view, our algorithm shows that
the  Dichromatic Number problem  is in $\xp$ when parameterized by directed clique-width.
Further, we show that  the Dichromatic Number problem is $\w[1]$-hard on symmetric digraphs when parameterized by directed clique-width.
Inferring from this, there is no $\fpt$-algorithm for the Dichromatic Number problem parameterized by directed clique-width under reasonable assumptions. The best parameterized complexity which can be achieved is given by an $\xp$-algorithm.
Furthermore, we apply defineability within monadic second order logic (MSO)
in order to show that Dichromatic Number problem in $\fpt$ when parameterized by the directed clique-width 
and $r$, which implies that for  every integer $r$ it holds that  $\DCN_{r}$
is in $\fpt$ when parameterized by directed clique-width.

Since the directed clique-width of a digraph is at most its directed modular width \cite{SW20},
we reprove the existence of an $\xp$-algorithm for $\DCN$
and an $\fpt$-algorithm for $\DCN_{r}$
parameterized by directed modular width \cite{SW19}.
On the other hand, there exist several classes of digraphs of bounded directed clique-width and
unbounded directed modular width, which implies that  directed clique-width is the more
powerful parameter and thus the results of \cite{SW19} does not imply any parameterized algorithm
for directed clique-width.

In Table \ref{tab} we summarize the known results for $\DCN$ and $\DCN_r$ parameterized by  parameters.

\begin{table}[h!]
\begin{center}
\begin{tabular}{l||ll|ll|}
 parameter                       & \multicolumn{2}{c|}{$\DCN$} &\multicolumn{2}{c|}{$\DCN_r$} \\
\hline
directed tree-width     & $\not\in\xp$ & Corollary \ref{cor-xp-ro}        & $\not\in\xp$   & Corollary \ref{cor-xp-ro}  \\
\hline
directed path-width     & $\not\in\xp$ &   Corollary \ref{cor-xp-ro}         & $\not\in\xp$  &Corollary \ref{cor-xp-ro}  \\
\hline
DAG-width               &$\not\in\xp$  &  Corollary \ref{cor-xp-ro}         & $\not\in\xp$  &Corollary \ref{cor-xp-ro}  \\
\hline
Kelly-width             & $\not\in\xp$ &  Corollary \ref{cor-xp-ro}         & $\not\in\xp$   &Corollary \ref{cor-xp-ro}   \\

\hline
directed modular width  & \fpt  & \cite{SW19}           & \fpt   &  \cite{SW19}   \\
\hline
directed clique-width   & \w[1]-hard & Corollary \ref{hpcw}   & \fpt   & Corollary \ref{fpt-dc}   \\
                        & \xp        & Corollary \ref{xp-dc}  &        &               \\
\hline
standard parameter  $r$       & $\not\in\xp$     & Corollary \ref{cor-xp-r} &  ///   &   \\
\hline
directed clique-width $+$ $r$   &  \fpt   &  Theorem    \ref{fpt-cw-r}               & ///       &    \\

\hline
clique-width of $\un(G)$     &  $\not\in\fpt$     &  Corollary \ref{cor-fpt-un}                     &  open    &       \\
\hline

number of vertices $n$   &  \fpt  &  Corollary \ref{fpt-n} & \fpt &  Corollary \ref{fpt-n} \\
\end{tabular}
\end{center}
\caption{Complexity of $\DCN$ and $\DCN_r$ parameterized by  parameters.
We assume that $\p\neq \np$. The ''///'' entries indicate that by taking $r$ out of the instance the considered parameter
makes no sense.
\label{tab}}
\end{table}

For directed co-graphs, which is a class of digraphs of directed clique-width 2 \cite{GWY16}, and several
generalizations we even show a linear time solution for computing the
dichromatic number and  an optimal acyclic coloring. Furthermore, we conclude that directed co-graphs and
the considered generalizations lead to subclasses of perfect digraphs  \cite{AH15}.
For  directed cactus forests, which is a set of digraphs of directed tree-width 1 \cite{GR19}, the 
results of  \cite{Wie20} and \cite{MSW19} lead to an upper bound of $2$ for the
dichromatic number and that an optimal acyclic coloring 
can be computed in polynomial time. We show that this even can be done in linear time.

\section{Preliminaries}\label{intro}

We use the notations of Bang-Jensen and Gutin \cite{BG18} for graphs and digraphs.

\subsection{Directed graphs}

A {\em directed graph} or {\em digraph} is a pair  $G=(V,E)$, where $V$ is
a finite set of {\em vertices} and
$E\subseteq \{(u,v) \mid u,v \in V,~u \not= v\}$ is a finite set of ordered pairs of distinct
vertices called {\em arcs} or {\em directed edges}.
For a vertex $v\in V$, the sets $N^+(v)=\{u\in V~|~ (v,u)\in E\}$ and
$N^-(v)=\{u\in V~|~ (u,v)\in E\}$ are called the {\em set of all successors}
and the {\em set of all  predecessors} of $v$.
The  {\em outdegree} of $v$, $\odeg(v)$ for short, is the number
of successors of $v$ and the  {\em indegree} of $v$, $\ideg(v)$ for short,
is the number of predecessors of $v$.

A digraph $G'=(V',E')$ is a {\em subdigraph} of digraph $G=(V,E)$ if $V'\subseteq V$
and $E'\subseteq E$.  If every arc of $E$ with both end vertices in $V'$  is in
$E'$, we say that $G'$ is an {\em induced subdigraph} of digraph $G$ and we
write $G'=G[V']$.

The {\em out-degeneracy} of a digraph $G$ is the least integer $d$
such that $G$ and all subdigraphs of $G$ contain a vertex of outdegree at most $d$.

For some given digraph $G=(V,E)$ we define
its {\em underlying undirected graph} by ignoring the directions of the arcs, i.e.
$\un(G)=(V,\{\{u,v\}~|~(u,v)\in E, u,v\in V\})$.
There are several ways to define a digraph $G=(V,E)$ from an undirected
graph $G'=(V,E')$.
If we replace every edge $\{u,v\}\in E'$ by
\begin{itemize}

\item
both arcs $(u,v)$ and $(v,u)$, we refer to $G$ as a {\em complete biorientation} of $G'$.
Since in this case $G$ is well defined by $G'$ we also denote
it by $\overleftrightarrow{G'}$.
Every digraph $G$  which can be obtained by a complete biorientation of some undirected
graph $G'$ is called a {\em complete bioriented graph} or {\em symmetric digraph}.

\item
one of the arcs $(u,v)$ and $(v,u)$, we refer to $G$ as an {\em orientation} of $G'$.
Every digraph $G$  which can be obtained by an orientation of some undirected
graph $G'$ is called an {\em oriented graph}.
\end{itemize}

For a digraph $G=(V,E)$ an arc $(u,v)\in E$ is {\em symmetric} if $(v,u)\in E$.
Thus, each bidirectional arc is symmetric. Further, an arc is {\em asymmetric} if it is not symmetric. 
We define the symmetric part of $G$ as $\sym(G)$,
which is the spanning subdigraph of $G$ that contains exactly the symmetric arcs of $G$.
Analogously, we define the asymmetric part of $G$ as $\asym(G)$, which is the spanning
subdigraph with only asymmetric arcs.

By $\overrightarrow{P_n}=(\{v_1,\ldots,v_n\},\{ (v_1,v_2),\ldots, (v_{n-1},v_n)\})$, $n \ge 2$,
we denote the directed path on $n$ vertices,
by $\overrightarrow{C_n}=(\{v_1,\ldots,v_n\},\{(v_1,v_2),\ldots, (v_{n-1},v_n),(v_n,v_1)\})$, $n \ge 2$,
we denote the directed cycle on $n$ vertices.

A {\em directed acyclic graph (DAG)} is a digraph without any $\overrightarrow{C_n}$,
for $n\geq 2$, as subdigraph. A vertex $v$ is {\em reachable} from a vertex $u$ in $G$ if $G$ contains
a $\overrightarrow{P_n}$ as a subdigraph having start vertex $u$ and
end vertex $v$.
%
%
%
A digraph is {\em odd cycle free} if it does not contain a $\overrightarrow{C_n}$,
for odd $n\geq 3$, as subdigraph.
A  digraph $G$ is planar if $\un(G)$ is planar.

A digraph  is  {\em even} if for every
0-1-weighting of the edges it contains a directed cycle of even total weight.

\subsection{Acyclic coloring of directed graphs}

We consider the approach for coloring digraphs given in \cite{NL82}.
A set $V'$ of vertices of a digraph $G$ is called {\em acyclic} if $G[V']$ is acyclic.

\begin{definition}[Acyclic graph coloring \cite{NL82}]\label{def-dc}
An \emph{acyclic $r$-coloring} of a digraph $G=(V,E)$ is a mapping $c:V\to \{1,\ldots,r\}$,
such that the color classes  $c^{-1}(i)$ for $1\leq i\leq r$ are acyclic.
The {\em dichromatic number} of $G$, denoted by $\vec{\chi}(G)$, is the smallest $r$,
such that $G$ has an acyclic $r$-coloring.
\end{definition}

There are several works on acyclic graph coloring  \cite{BFJKM04,Moh03,NL82} including several
recent works~\cite{LM17,MSW19,SW20}. The following observations support that the dichromatic number
can be regarded as a natural counterpart of the well known
chromatic number $\chi(G)$ for undirected graphs $G$.

\begin{observation}\label{obs-dicol}
For every symmetric directed graph $G$ it holds that
$\vec{\chi}(G) = \chi(\un(G))$.
\end{observation}

\begin{observation} For every directed graph $G$ it holds that
$\vec{\chi}(G)\leq \chi(\un(G))$.
\end{observation}

\begin{observation}\label{le-ubdigraph}
Let $G$ be a digraph  and $H$ be a subdigraph
of $G$, then $\vec{\chi}(H)\leq \vec{\chi}(G)$.
\end{observation}

We consider the following problem.

\begin{desctight}
\item[Name] Dichromatic Number ($\DCN$)

\item[Instance]  A digraph $G=(V,E)$ and a positive integer $r \leq |V|$.

\item[Question]  Is there an acyclic $r$-coloring for $G$?
\end{desctight}

If $r$ is constant and not part of the input, the corresponding problem
is denoted by $r$-Dichromatic Number ($\DCN_{r}$). Even  $\DCN_{2}$ is NP-complete  \cite{FHM03}.

\section{Acyclic coloring of special digraph classes}

As recently mentioned in \cite{SW19}, only few classes of digraphs for which the dichromatic number
can be found in polynomial time are known. 

\begin{observation}
The set of DAGs is equal to the set of digraphs
of  dichromatic  number  $1$. 
\end{observation}

\begin{proposition}[\cite{NL82}]
 Every odd-cycle free digraph   has  dichromatic number at most $2$.
\end{proposition}

\begin{proposition}[\cite{MSW19}]
 Every  non-even digraph  has  dichromatic number at most $2$.
\end{proposition}

Thus, given a DAG its  dichromatic number is always equal to $1$. 
Further, for some given odd-cycle free digraph or non-even digraph
we first check in linear time whether the  digraph is a DAG and thus has 
dichromatic number $1$ or otherwise state that it has dichromatic number $2$.

\begin{corollary}
For DAGs, odd-cycle free digraphs, and   non-even digraphs 
the dichromatic number can be computed in linear time.
\end{corollary}

%

\subsection{Acyclic coloring of perfect digraphs}

The dichromatic number and the clique number are the two basic concepts
for the class of perfect digraphs~\cite{AH15}.
The {\em clique number} of a digraph $G$, denoted by  $\omega_d(G)$, is
the number of vertices in a largest complete bioriented subdigraph of $G$.

\begin{definition}[Perfect digraphs \cite{AH15}]\label{defp}
A digraph $G$ is {\em perfect} if, for every induced subdigraph $H$ of $G$, the dichromatic
number $\vec{\chi}(H)$ equals the clique number $\omega_d(H)$.
\end{definition}

An undirected graph $G$ is perfect if and only if its complete biorientation  $\overleftrightarrow{G}$
is a perfect digraph. Thus, Definition \ref{defp} is a  generalization of perfectness to digraphs.
While for undirected perfect graphs more than a hundred subclasses have been defined and studied \cite{Hou06},
for perfect digraphs there are only very few subclasses known. Obviously, DAGs and subclasses
such as series-parallel digraphs, minimal series-parallel digraphs and series-parallel order
digraphs \cite{VTL82} are perfect digraphs.

By \cite{AH15} the dichromatic number of a perfect digraph $G$
can be found by the chromatic number of $\un(\sym(G))$,
which is possible in polynomial time \cite{GLS81}.

\begin{proposition}[\cite{AH15}]\label{perf-col}
For every perfect digraph the Dichromatic Number problem can be solved in polynomial time.
\end{proposition}

We show how to find an optimal
acyclic coloring for directed co-graphs and special digraphs of directed 
tree-width one in linear time.

\subsection{Acyclic coloring of directed co-graphs}

Let $G_1=(V_1,E_1)$ and $G_2=(V_2,E_2)$ be two vertex-disjoint digraphs.
The following opera\-tions
have been considered  by Bechet et al.\ in \cite{BGR97}.
\begin{itemize}
\item
The {\em disjoint union} of $G_1$ and $G_2$,
denoted by $G_1 \oplus  G_2$,
is the digraph with vertex set $V_1\cup  V_2$ and
arc set $E_1\cup  E_2$.

\item
The {\em series composition} of $G_1$ and $G_2$,
denoted by $G_1\otimes G_2$,
is defined by their disjoint union plus all possible arcs between
vertices of $G_1$ and $G_2$.

\item
The {\em order composition} of $G_1$ and $G_2$,
denoted by $G_1\oslash  G_2$,
is defined by their disjoint union plus all possible arcs from
vertices of $G_1$ to vertices of $G_2$.
\end{itemize}

The following transformation has  been considered  by Johnson et al.\ in \cite{JRST01}
and generalizes the operations disjoint union and order composition.
\begin{itemize}
\item
A graph $G$ is obtained by a {\em directed union} of
$G_1$ and $G_2$, denoted by  $G_1 \ominus G_2$ if  $G$ is a
 subdigraph of the order composition of $G_1 \oslash  G_2$
 and contains the disjoint union $G_1 \oplus  G_2$ as a subdigraph.
\end{itemize}

We recall the definition of directed co-graphs from \cite{CP06}.

\begin{definition}[Directed co-graphs \cite{CP06}]\label{dcog}
The class of {\em directed co-graphs} is recursively defined as follows.
\begin{enumerate}
\item Every digraph with a single vertex $(\{v\},\emptyset)$,
denoted by $v$, is a {\em directed co-graph}.

\item If  $G_1$ and $G_2$  are vertex-disjoint directed co-graphs, then
\begin{enumerate}
\item
the disjoint union
$G_1\oplus G_2$,

\item
the series composition
$G_1 \otimes G_2$, and
\item
the order composition
$G_1\oslash G_2$  are {\em directed co-graphs}.
\end{enumerate}
\end{enumerate}
\end{definition}

Every expression $X$ using  the four operations of Definition  \ref{dcog}
is called a {\em di-co-expression} and
$\g(X)$ is the defined digraph.

As undirected co-graphs can be characterized by forbidding the $P_4$,
directed co-graphs can be characterized likewise by
excluding the eight forbidden induced subdigraphs \cite{CP06}.

For every directed co-graph we can define a tree structure
denoted as {\em di-co-tree}. It is an ordered rooted tree whose
leaves represent the vertices of the digraph and
whose inner nodes correspond
to the operations applied on the subexpressions defined by the subtrees.
For every directed co-graph one can construct a di-co-tree in linear time \cite{CP06}.
Directed co-graphs are interesting from an algorithmic point of view
since several hard graph problems can be solved in
polynomial time  by dynamic programming along the tree structure of
the input graph, see \cite{BM14,Gur17a,GHKRRW20,GKR19f,GKR19d,GR18c,Ret98}.  
Moreover, directed co-graphs are very useful for the reconstruction
of the evolutionary history of genes or species using genomic
sequence data \cite{HSW17,NEMWH18}.

The set of {\em extended directed co-graphs}
was introduced in \cite{GR18f} by allowing
the directed union and series composition of defined digraphs which leads
to a superclass of directed co-graphs.
Also for every extended directed co-graph we can define a tree structure,
denoted as {\em ex-di-co-tree}.
For the class of extended directed co-graphs it remains open how to compute an
ex-di-co-tree.

\begin{lemma}\label{le-dec} Let $G_1$ and $G_2$ be two vertex-disjoint directed graphs.
Then, the following equations hold:
\begin{enumerate}
\item $\vec{\chi}((\{v\},\emptyset))=1$

\item $\vec{\chi}(G_1 \oplus G_2) = \vec{\chi}(G_1\oslash  G_2) = \vec{\chi}(G_1\ominus  G_2) = \max(\vec{\chi} (G_1), \vec{\chi}(G_2))$

\item $\vec{\chi}(G_1\otimes  G_2)= \vec{\chi}(G_1) +  \vec{\chi}(G_2)$
\end{enumerate}
\end{lemma}

\begin{proof}
\begin{enumerate}
\item
$\vec{\chi}((\{v\},\emptyset)) = 1$ is obviously clear.


\item  $\vec{\chi}(G_1 \oplus G_2) \geq \max(\vec{\chi} (G_1), \vec{\chi}(G_2))$

Since the digraphs $G_1$ and  $G_2$ are induced subdigraphs of digraph $G_1 \oplus  G_2$,
both values $\vec{\chi} (G_1)$ and $\vec{\chi}(G_2)$ lead to
a lower bound for the number of necessary colors
of the combined graph by Observation~\ref{le-ubdigraph}.

$\vec{\chi}(G_1 \oplus  G_2) \leq \max(\vec{\chi} (G_1), \vec{\chi}(G_2))$

Since the disjoint union operation does not create any new arcs, we
can combine color classes of $G_1$ and $G_2$.

The results for the order composition and directed union follow by the same arguments.


\item $\vec{\chi}(G_1\otimes  G_2) \geq \vec{\chi}(G_1) +  \vec{\chi}(G_2)$

Since every $G_1$ and $G_2$ are an induced subdigraphs of the combined
graph, both values $\vec{\chi} (G_1)$ and $\vec{\chi}(G_2)$ lead to
a lower bound for the number of necessary colors
of the combined graph by Observation~\ref{le-ubdigraph}. Further,
the series operations implies that
every vertex in $G_1$ is on a cycle of length two with every vertex
of $G_2$. Thus, no vertex in $G_1$ can be
colored in the same way as a vertex in $G_2$.
Thus, $\vec{\chi}(G_1) + \vec{\chi}(G_2)$
leads to a lower bound for the number of necessary colors
of the combined graph.

$\vec{\chi}(G_1\otimes  G_2) \leq \vec{\chi}(G_1) +\vec{\chi}(G_2)$

For $1\leq i \leq 2$ let $G_i=(V_i,E_i)$ and $c_i:V_i\to\{1,\ldots,\vec{\chi}(G_i)\}$
a coloring for $G_i$.
For $G_1\otimes    G_2=(V,E)$ we define a
mapping $c:V\to \{1,\ldots, \vec{\chi}(G_1)+\vec{\chi}(G_2) \}$ as follows.
\[
c(v)=\left\{ \begin{array}{ll}
c_1(v)                        &  {\rm\ if\ } v \in V_1 \\

c_2(v)+ \vec{\chi}(G_1)   &  {\rm\ if\ } v \in V_{2}. \\
\end{array}\right.
\]

The mapping $c$ satisfies the definition of an acyclic coloring, because
every color class $c^{-1}(j)$, $j \in\{1,\ldots, \vec{\chi}(G_1)+\vec{\chi}(G_2)  \}$ is a subset of
$V_1$ or of $V_2$, such that $c^{-1}(j)$
induces an acyclic digraph in $G_1$ or $G_2$ by assumption.
Since the series operation does not insert any further arcs between two vertices
of $G_1$ or $G_2$, vertex set $c^{-1}(j)$ induces also an  acyclic digraph in $G$.
\end{enumerate}
This shows the statements of the lemma.
\end{proof}

Lemma \ref{le-dec} can be used to obtain the following result.

\begin{theorem}\label{algopd}
Let $G$ be a (-n extended) directed co-graph (given by an ex-di-co-tree).
Then, an optimal acyclic coloring for $G$ and
$\vec{\chi}(G)$  can be computed in linear time.
\end{theorem}

In order to state the next result, let $\omega(G)$
be the number of vertices in a largest clique in graph $G$.
Since the results of Lemma \ref{le-dec} also hold for  $\omega_d$ instead of $\vec{\chi}$
we obtain the following result.

\begin{proposition}\label{prop-p}
Let $G$ be a  (-n extended)  directed co-graph (given by an ex-di-co-tree). Then, it holds that
$\vec{\chi}(G)=\chi(\un(\sym(G)))=\omega(\un(\sym(G)))=\omega_d(G)$ and all values can be
computed in linear time.
\end{proposition}

\begin{proposition}\label{dicop}
Every (extended) directed co-graph is a perfect digraph.
\end{proposition}

\begin{proof}
We show the result by verifying Definition \ref{defp}.
Since every induced subdigraph of a (-n extended) directed co-graph
is a (-n extended) directed co-graph, Proposition \ref{prop-p}
implies that every (extended) directed co-graph is  a perfect digraph.
\end{proof}

Alternatively,
the last result can be shown by the {\em Strong Perfect Digraph Theorem} \cite{AH15}
since for every (extended) directed co-graph $G$ the symmetric part $\un(\sym(G))$ is an
undirected co-graph and thus a perfect graph. Furthermore (extended) directed co-graphs
do not contain a directed cycle $\overrightarrow{C_n}$, $n\geq 3$,
as an induced subdigraph \cite{CP06,GKR19h}.


The results of Theorem \ref{algopd} and Propositions  \ref{prop-p} and \ref{dicop} can be generalized
to larger classes. Motivated by the idea of tree-co-graphs \cite{Tin89} we
replace in Definition \ref{dcog} the single vertex graphs by a DAG, for which we
know that the dichromatic number and
also the clique number is equal to $1$. Thus, Lemma \ref{le-dec} can be
adapted to compute the  dichromatic number in linear time. Furthermore
following the proof of Proposition~\ref{dicop} we know that this class
is perfect.

\subsection{Acyclic coloring of directed cactus forests}

We recall the definition of directed cactus forests, which was introduced in \cite{GR19}
as a counterpart for undirected cactus forests.

\begin{definition}[Directed cactus forests \cite{GR19}]\label{def-tl}
A \emph{directed cactus forest} is a digraph, where any two directed cycles
have at most one joint vertex.
\end{definition}



Since directed cactus forests may contain a cycle they have
dichromatic number at least  $2$. 
Further, the set of all cactus forests is a subclass of the digraphs of directed tree-width at most $1$ \cite{GR19}, which
is a subclass of non-even digraphs \cite{Wie20}. By \cite{MSW19} we can conclude that every cactus forest
has dichromatic number at most $2$ and that
for every cactus forest an optimal acyclic coloring 
can be computed in polynomial time. In order
to improve the running time, we show the following lemma.

\begin{lemma}\label{cf2}
Let $G$ be a  directed cactus forest. Then, an acyclic 2-coloring for $G$ and
can be computed in linear time.
\end{lemma}

\begin{proof}
Let $G=(V,E)$ be a directed cactus forest.
In order to define an acyclic 2-coloring $c:V\to \{1,2\}$ for $G$,
we traverse a DAG $D_G$ defined as follows.
$D_G$ has a vertex for every  cycle $C$ on at least two vertices  in $G$ and
a vertex for every vertex of $G$ which is not
involved in any such cycle of $G$.

Two vertices
in $D_G$  which both represent a  cycle in $G$ are
adjacent by a single (arbitrary chosen)
directed edge in $D_G$ if these cycles have a common vertex in $G$.
A vertex  in $D_G$ which represents a  cycle $C$ in $G$ and
a vertex  in $D_G$ which represents a vertex $u$ in $G$ which
is not involved in any   cycle in $G$ are adjacent in the same
way as the vertex of $C$ is adjacent to $u$ in $G$.
Two vertices in $D_G$  which both represent a vertex in $G$ which
is not involved in any  cycle in $G$ are
adjacent in $D_G$ in the same way as they are adjacent in $G$.

Then, we consider the vertices $v$ of $D_G$ by a topological ordering
in order to define an acyclic 2-coloring for $G$.
\begin{itemize}

\item If $v$  represents a vertex in $G$ which
is not involved in any cycle in $G$, we define $c(v)=2$.

\item If $v$ represents a cycle $C$ in $G$, then we distinguish the following cases.

\begin{itemize}
\item If all vertices of $C$ are unlabeled up to now, we choose an arbitrary vertex
$x\in C$ and define $c(x)=1$ and $c(y)=2$ for all $y\in C-\{x\}$.

\item If we have already labeled one vertex $x$ of $C$ by $1$ then
we define $c(y)=2$ for all $y\in C-\{x\}$.

\item If we have already labeled one vertex $x$ of $C$ by $2$ then we choose an arbitrary vertex
$x'\in C-\{x\}$ and define $c(x')=1$ and further $c(y)=2$ for all $y\in C-\{x,x'\}$.
\end{itemize}
\end{itemize}
By this definitions, in every  cycle $C$ of $G$ we color exactly one vertex of $C$ by $1$
and all remaining vertices of $G$ by $2$.
Thus, $c$ leads to a legit acyclic 2-coloring for $G$.
\end{proof}

 If $G$ is a DAG, then $c(x)=1$ for
$x\in V$ leads to an  acyclic 1-coloring for $G$. Otherwise, Lemma \ref{cf2} leads to an
 acyclic 2-coloring for $G$.

\begin{theorem}\label{algocf}
Let $G$ be a  directed cactus forest. Then, an optimal acyclic coloring for $G$ and
$\vec{\chi}(G)$  can be computed in linear time.
\end{theorem}

\section{Parameterized algorithms for directed clique-width}

For undirected graphs the clique-width \cite{CO00} is one of the most important parameters.
Clique-width measures how difficult it is to decompose the graph into a special tree-structure. 
From an algorithmic point of view, 
only tree-width \cite{RS86} is a more studied graph parameter.
Clique-width is more general than  tree-width since
graphs of bounded tree-width have also bounded clique-width \cite{CR05}.
The tree-width can only be bounded by the clique-width under
certain conditions \cite{GW00}.
Many NP-hard graph problems admit poly\-nomial-time solutions when restricted to
graphs of bounded tree-width or graphs of bounded clique-width.

For directed graphs there are several attempts
to generalize tree-width 
such as directed path-width, directed tree-width, DAG-width, or Kelly-width, which
are representative for what people are working on, see  the surveys
\cite{GHKLOR14,GHKMORS16}. Unfortunately, none of these attempts
allows polynomial-time algorithms for a large class of problems on
digraphs of bounded width. This also holds for $\DCN_r$ and $\DCN$ by the following
theorem.

\begin{theorem}[\cite{MSW19}]
For every  $r \geq 2$ the $r$-Dichromatic Number problem is NP-hard even for input digraphs
whose feedback vertex set number is
at most $r+4$ and whose out-degeneracy is at most $r+1$.
\end{theorem}

Thus, even for bounded size of a directed feedback vertex set,
deciding whether a directed graph has dichromatic number at most 2 is NP-complete.
This result rules out $\xp$-algorithms for $\DCN$ and $\DCN_r$ by directed width
parameters such as directed path-width, directed tree-width, DAG-width or Kelly-width, 
since all of these are upper bounded by the feedback vertex set number.

\begin{corollary}\label{cor-xp-ro}
The Dichromatic Number problem  is not in $\xp$ when parameterized by directed 
tree-width, directed path-width, Kelly-width, or DAG-width, unless  $\p=\np$.
\end{corollary}

Next, we discuss parameters which allow
$\xp$-algorithms or even $\fpt$-algorithms for $\DCN$ and $\DCN_r$.
The first positive result concerning structural parameterizations of $\DCN$
was recently given in \cite{SW19} using the directed modular width ($\dmws$).

\begin{theorem}[\cite{SW19}]\label{fpt-mw}
The  Dichromatic Number problem is in $\fpt$ when parameterized by directed modular width.
\end{theorem}

By  \cite{GHKLOR14}, directed
clique-width performs much better than directed path-width, directed tree-width, DAG-width, and
Kelly-width from the parameterized complexity point of view.
Hence, we consider the parameterized complexity of $\DCN$   parameterized by
directed clique-width. 

\begin{definition}[Directed clique-width \cite{CO00}]\label{D4}
The {\em directed clique-width} of a digraph $G$, $\dcws(G)$ for short,
is the minimum number of labels needed to define $G$ using the following four operations:

\begin{enumerate}
\item Creation of a new vertex $v$ with label $a$ (denoted by $a(v)$).

\item Disjoint union of two labeled digraphs $G$ and $H$ (denoted by $G\oplus H$).

\item Inserting an arc from every vertex with label $a$ to every vertex with label $b$
($a\neq b$, denoted by $\alpha_{a,b}$).

\item Change label $a$ into label $b$ (denoted by $\rho_{a\to b}$).
\end{enumerate}
An expression $X$ built with the operations defined above 
using $k$ labels is called a {\em directed  clique-width $k$-expression}.
Let $\g(X)$ be the digraph defined by $k$-expression $X$.
\end{definition}
%


In \cite{GWY16} the set of directed co-graphs is characterized by excluding two
digraphs  as a proper subset
of the set of all graphs of directed clique-width 2, while for the  undirected versions
both classes are equal.

By the given definition every graph of directed clique-width at most $k$ can be represented by a tree structure,
denoted as {\em $k$-expression-tree}. The leaves of the  $k$-expression-tree represent the
vertices of the digraph and the inner nodes of the  $k$-expression-tree  correspond
to the operations applied to the subexpressions defined by the subtrees.
Using the  $k$-expression-tree many hard problems have been shown to be
solvable in polynomial time when restricted to graphs of bounded directed clique-width \cite{GWY16,GHKLOR14}.

Directed clique-width is not comparable to the directed variants of tree-width mentioned above, 
which can be observed by the set of all complete biorientations of cliques and the set of all acyclic
orientations of grids.
The relation of directed clique-width and directed modular width  \cite{SW20} is as follows.

\begin{lemma}[\cite{SW20}] For every digraph $G$ it holds that
$\dcws(G)\leq \dmws(G)$.
\end{lemma}

On the other hand, there exist several classes of digraphs of bounded directed clique-width and
unbounded directed modular width, e.g. even the set of all
directed paths $\{\overrightarrow{P_n} \mid n\geq 1\}$,
the set of all directed cycles $\{\overrightarrow{C_n} \mid n\geq 1\}$, and the set of all minimal series-parallel digraphs \cite{VTL82}.
Thus, the result of \cite{SW19} does not imply any $\xp$-algorithm or $\fpt$-algorithm
for directed clique-width.

\begin{corollary}\label{hpcw}
The Dichromatic Number problem is $\w[1]$-hard on symmetric digraphs and thus, on all digraphs when parameterized by
directed clique-width.
\end{corollary}

\begin{proof}
The Chromatic Number problem
is $\w[1]$-hard parameterized by clique-width \cite{FGLS10a}.
An instance consisting of a graph $G=(V,E)$ and a positive integer $r$ for the
Chromatic Number problem can be transformed
into an instance for the Dichromatic Number problem
on digraph  $\overleftrightarrow{G}$ and integer $r$.
Then, $G$ has an  $r$-coloring if and only if $\overleftrightarrow{G}$ has an
acyclic $r$-coloring by Observation \ref{obs-dicol}.
Since for every undirected graph $G$ its clique-width equals the directed
clique-width of $\overleftrightarrow{G}$ \cite{GWY16}, we obtain a
parameterized reduction.
\end{proof}

Thus, under reasonable assumptions there is no $\fpt$-algorithm for the Dichromatic Number problem
parameterized by directed clique-width and an $\xp$-algorithm is the best that can be achieved.
Next, we introduce such an  $\xp$-algorithm.

Let $G=(V,E)$ be a digraph which is given by some directed clique-width $k$-expression $X$.
For some vertex set $V'\subseteq V$, we define $\reach(V')$
as the set of all pairs $(a,b)$ such that there is a vertex $u\in V'$
labeled by $a$  and there is a vertex $v\in V'$
labeled by $b$ and $v$ is reachable from $u$ in $G[V']$.

\begin{example}\label{ex-f}
We consider the digraph in Figure \ref{F06}.
The given partition into three acyclic sets $V=V_1\cup V_2 \cup V_3$, where $V_1=\{v_1,v_6,v_7\}$,
$V_2=\{v_2\}$, and $V_3= \{v_3,v_4,v_5\}$ leads to the multi set
${\cal M}=\langle \reach(V_1),\reach(V_2), \reach(V_3) \rangle$, where
$\reach(V_1)=\reach(V_3)= \{(1,1),(2,2),(4,4),(1,2),(2,4),(1,4)\}$ and $\reach(V_2) =\{(3,3)\}$.
\end{example}

\begin{figure}[hbtp]
\centerline{\includegraphics[width=0.37\textwidth]{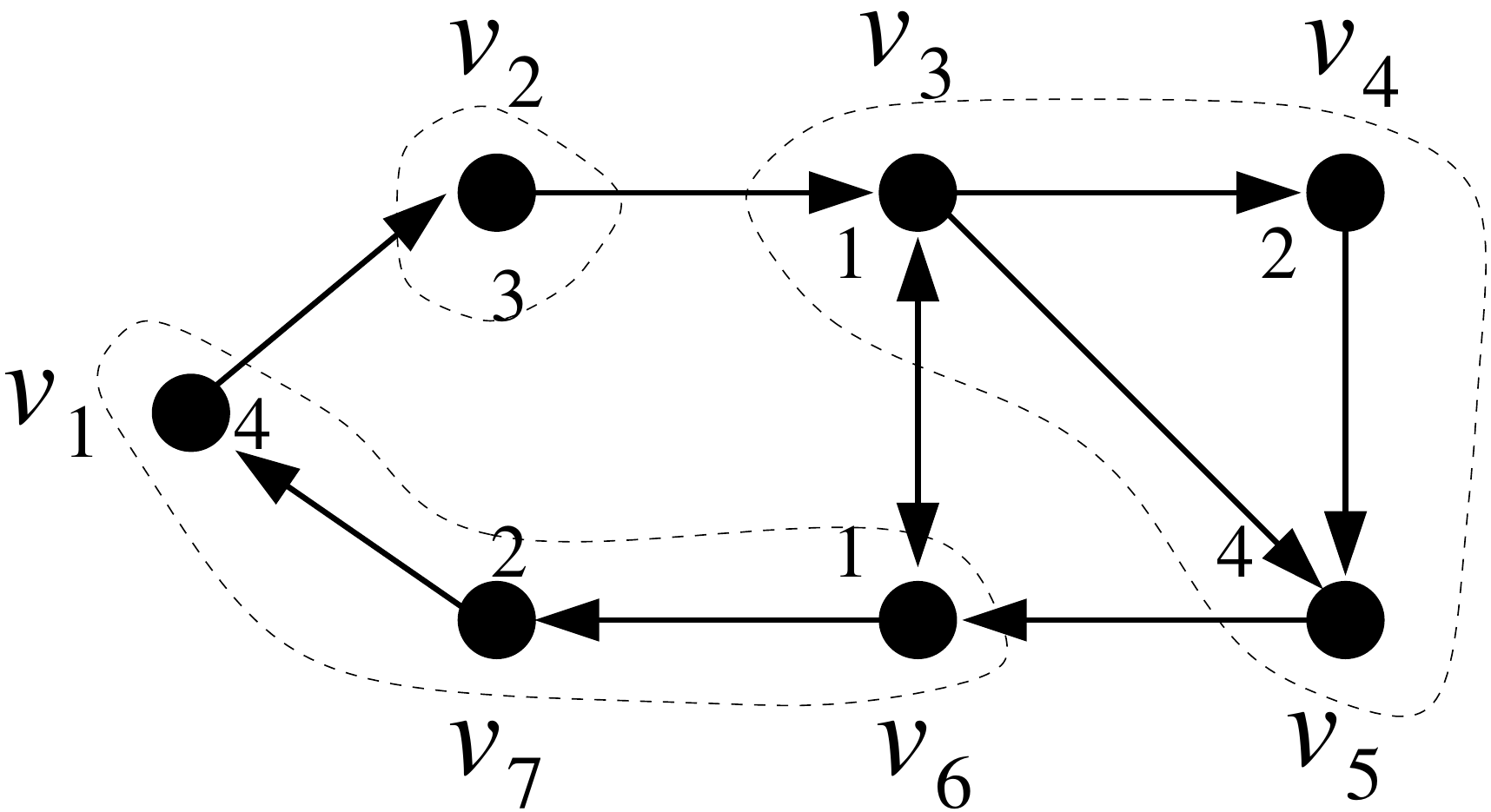}}
\caption{Digraph in Example \ref{ex-f}. The dashed lines indicate a partition of
the vertex set into three acyclic sets. The numbers at the vertices represent their
labels.
\label{F06}}
\end{figure}

Within a construction of a digraph by directed clique-width operations
only the edge insertion operation can change the reachability between
the present vertices. Next, we show which acyclic sets remain acyclic
when performing an edge insertion operation and how the
reachability information of these sets have to be updated
due to the edge insertion operation.

\begin{lemma}\label{le0}
Let $G=(V,E)$ be a vertex labeled digraph defined by some  directed clique-width $k$-expression $X$,
$a\neq b$, $a,b \in\{1,\ldots,k\}$,
and $V'\subseteq V$ be an acyclic set in $G$. Then, vertex set $V'$ remains acyclic in $\g(\alpha_{a,b}(X))$
if and only if $(b,a)\not\in \reach(V')$.
\end{lemma}

\begin{proof}
If $(b,a)\in \reach(V')$, then we know that in $\g(X)$ there is a vertex $y$
labeled by $a$ which is reachable by a vertex $x$ labeled by $b$. That is,
in $\g(X)$ there is a directed path $P$ from $x$ to $y$.
The edge insertion $\alpha_{a,b}$ leads to the arc $(y,x)$ which leads together with
path $P$ to a cycle in  $\g(\alpha_{a,b}(X))$.

If $(b,a)\not\in \reach(V')$ and $V'\subseteq V$ is an acyclic set in $\g(X)$,
then there is a topological ordering of $\g(X)[V']$
such that every vertex labeled by $a$ is before every vertex labeled by $b$ in the ordering.
The same ordering is  a topological ordering for $\g(\alpha_{a,b}(X))[V']$ which implies that
$V'$ remains acyclic for $\g(\alpha_{a,b}(X))$.
\end{proof}

\begin{lemma}\label{le0x}
Let $G=(V,E)$ be a vertex labeled digraph defined by some  directed clique-width $k$-expression $X$,
$a\neq b$, $a,b \in\{1,\ldots,k\}$, $V'\subseteq V$ be an acyclic set in $G$, and $(b,a)\not\in \reach(V')$.
Then, $\reach(V')$ for $\g(\alpha_{a,b}(X))$ can be obtained
from $\reach(V')$ for  $\g(X)$ as follows:
\begin{itemize}
\item For every pair $(x,a)\in \reach(V')$  and every pair $(b,y)\in \reach(V')$, we  extend $\reach(V')$  by $(x,y)$.
\end{itemize}
\end{lemma}

\begin{proof}
Let $R_1$ be the set $\reach(V')$ for  $\g(X)$, $R_2$ be the set $\reach(V')$ for $\g(\alpha_{a,b}(X))$,
and $R$ be the set of pairs constructed in the lemma starting with $\reach(V')$ for  $\g(X)$.
Then, it holds $R_1\subseteq R$. Furthermore, the rule given in the lemma
obviously puts feasible pairs into $\reach(V')$
which implies $R\subseteq R_2$. It remains to show $R_2\subseteq R$.
Let $(c,d)\in R_2$. If $(c,d)\in R_1$ then $(c,d)\in R$ as mentioned
above. Thus, let $(c,d)\not\in R_1$. This implies that there is a vertex $u\in V'$
labeled by $c$  and  a vertex $v\in V'$
labeled by $d$  and $v$ is reachable from $u$ in $\g(\alpha_{a,b}(X))$.
Since  $\g(X)$ is a spanning subdigraph of $\g(\alpha_{a,b}(X))$ and
the vertex labels are not changed by the performed edge insertion operation,
there is a non-empty set $V_u\in V'$ of vertices
labeled by $c$  and there is a non-empty set $V_v\in V'$
of vertices  labeled by $d$ and no vertex of $V_v$ is reachable from
a vertex of $V_u$ in   $\g(X)$. By the definition of the edge
insertion operation we know that  in  $\g(X)$ there is a
vertex $u'$ labeled by $a$ and a vertex $v'$ labeled by $b$ such that
$u'$ is reachable from $u$ and $v$ is reachable from $v'$. Thus,
$(c,a)\in R_1$ and $(b,d)\in R_1$. Our rule given in the statement
leads to $(c,d)\in R$.
\end{proof}

For a disjoint partition of $V$ into acyclic
sets $V_1,\ldots,V_s$, let  ${\cal M}$ be the multi
set\footnote{We use the notion of a {\em multi set}, i.e., a set
that may have several equal elements. For a multi set with elements $x_1,\ldots,x_n$ we write
${\cal M}=\langle x_1,\ldots,x_n \rangle$.
There is no order on the elements of ${\cal M}$.
The number how often an element $x$ occurs in ${\cal M}$ is denoted by $\psi({\cal M},x)$.
Two multi sets ${\cal M}_1$ and ${\cal M}_2$
are {\em equal} if for each element $x \in {\cal M}_1 \cup {\cal M}_2$,
$\psi({\cal M}_1,x)=\psi({\cal M}_2,x)$, otherwise they are called
{\em different}. The empty multi set is denoted by $\langle \rangle$.}
$\langle \reach(V_1),\ldots,\reach(V_s) \rangle$.

Let $F(X)$ be the set of all mutually different multi sets
${\cal M}$ for all disjoint partitions of vertex set $V$ into acyclic sets.
Every multi set in $F(X)$ consists of  nonempty subsets of
$\{1,\ldots,k\} \times \{1,\ldots,k\}$. Each subset can occur
$0$ times and not more than $|V|$ times.
Thus, $F(X)$ has at most $$(|V|+1)^{2^{k^2}-1} \in  |V|^{2^{\bigo(k^2)}} $$ mutually different multi
sets
and is polynomially bounded in the size of $X$.

In order to give a dynamic programming solution along the recursive structure
of a directed clique-width $k$-expression, we show how
to compute $F(a(v))$, $F(X \oplus Y)$  from $F(X)$
and $F(Y)$, as well as $F(\alpha_{a,b}(X))$ and $F(\rho_{a \to b}(X))$ from $F(X)$.

\begin{lemma}\label{le1} Let $a,b\in\{1,\ldots,k\}$, $a\neq b$.
\begin{enumerate}
\item
$F(a(v))=\{ \langle \{(a,a)\} \rangle \}$.

\item
Starting with set $D=\{ \langle \rangle \} \times F(X) \times F(Y)$
extend $D$ by all triples that can be obtained from some
triple $({\cal M},{\cal M}',{\cal M}'') \in D$ by removing a set $L'$ from ${\cal M}'$
or a set $L''$ from ${\cal M}''$ and inserting it into ${\cal M}$, or
by removing both sets and inserting $L' \cup L''$ into ${\cal M}$.
Finally, we choose
$F(X \oplus Y)= \{ {\cal M} \mid  ({\cal M},\langle \rangle,\langle \rangle) \in D\}$.

\item
$F(\alpha_{a,b}(X))$ can be obtained from  $F(X)$ as follows.
First, we remove from  $F(X)$ all multi sets $\langle L_1, \ldots, L_s \rangle$
such that  $(b,a) \in L_t$ for some $1\leq t\leq s$.
Afterwards, we modify every remaining multi set
$\langle L_1, \ldots, L_s \rangle$ in  $F(X)$ as follows:
\begin{itemize}
\item
For every $L_i$ which contains
a pair $(x,a)$ and a pair $(b,y)$, we  extend $L_i$ by $(x,y)$.
\end{itemize}

\item
$F(\rho_{a \to b}(X)) = \{ \langle \rho_{a \to b}(L_1), \ldots, \rho_{a \to b}(L_s) \rangle
\mid  \langle L_1,\ldots,L_s \rangle \in F(X) \}$, using the relabeling of a set of label pairs  $\rho_{a\to b}(L_i)=\{(\rho_{a\to b}(c),\rho_{a\to b}(d)) \mid (c,d)\in L_i\}$  and the relabeling
of integers $\rho_{a\to b}(c)=b$, if $c=a$ and $\rho_{a\to b}(c)=c$, if $c\neq a$.
\end{enumerate}
\end{lemma}

\begin{proof}
\begin{enumerate}
\item In $\g(a(v))$ there is exactly one vertex $v$ labeled by $a$ and thus the only partition of $V$ into
one acyclic set of the vertex set of $\g(a(v))$ is $\{v\}$. The corresponding multi set is
$\langle \reach(\{v\}) \rangle=\langle \{(a,a)\} \rangle$.

\item
$F(X \oplus Y) \subseteq \{ {\cal M} \mid  ({\cal M},\langle \rangle,\langle \rangle) \in D\}$:

Every acyclic set of  $\g(X \oplus Y)$ is either an acyclic set in $\g(X)$, or an
acyclic set in $\g(Y)$, or is the union of two acyclic sets from $\g(X)$ and $\g(Y)$.
All three possibilities are considered when computing $\{ {\cal M} \mid  ({\cal M},\langle \rangle,\langle \rangle) \in D\}$ from $F(X)$ and $F(Y)$.

$F(X \oplus Y) \supseteq \{ {\cal M} \mid  ({\cal M},\langle \rangle,\langle \rangle) \in D\}$:

Since the operation $\oplus$ does not create any new edges, the acyclic sets
from $\g(X)$, the acyclic sets
from $\g(Y)$, and the union of acyclic sets from $\g(X)$ and $\g(Y)$ remain acyclic sets for $\g(X \oplus Y)$.

\item By Lemma \ref{le0} we have to remove all multi sets $\langle L_1, \ldots, L_s \rangle$  from  $F(X)$
for which holds that  $(b,a) \in L_t$ for some $1\leq t\leq s$. The remaining multi sets are updated
correctly by Lemma \ref{le0x}.

\item In $\g(X)$ there is a vertex labeled by $d$ which is reachable
from a vertex  labeled by $c$ if and only if in
$\g(\rho_{a \to b}(X))$ there is a vertex labeled by $\rho_{a \to b}(d)$ which is reachable
from a vertex labeled by $\rho_{a\to b}(c)$.
\end{enumerate}
This shows the statements of the lemma.
\end{proof}

Since every possible coloring of $G$ is realized in the set $F(X)$, where $X$ is a
directed clique-width $k$-expression for 
$G$, it is easy to find a minimum coloring for $G$.

\begin{corollary} \label{cor3}
Let $G=(V,E)$ be a digraph given by a directed clique-width $k$-expression $X$.
There is a partition of $V$ into $r$  acyclic sets
if and only if there is some ${\cal M} \in F(X)$ consisting of $r$ sets of label pairs.
\end{corollary}

\begin{theorem}\label{xp-dca}
The Dichromatic Number problem on digraphs on $n$ vertices given by a directed
clique-width $k$-expression can be solved
in $n^{2^{\bigo(k^2)}}$ time.
\end{theorem}

\begin{proof}
Let $G=(V,E)$  be a digraph of directed clique-width at most $k$ and $T$ be a $k$-expression-tree
for $G$ with root $w$.
For some vertex $u$ of $T$ we denote by $T_u$
the subtree rooted at $u$ and $X_u$ the $k$-expression defined by $T_u$.
In order to solve the Dichromatic Number problem for $G$,
we traverse $k$-expression-tree  $T$ in a bottom-up order.
For every vertex $u$ of $T$ we compute $F(X_u)$ following the rules
given in Lemma \ref{le1}. By Corollary \ref{cor3} we can solve our
problem by $F(X_w)=F(X)$.

Our rules given Lemma \ref{le1} show the following running times.
For every $v\in V$ and  $a\in\{1,\ldots,k\}$ set $F(a(v))$
can be computed in $\bigo(1)$.
The set $F(X \oplus Y)$ can be computed
in time $(n+1)^{3(2^{k^2}-1)}\in n^{2^{\bigo(k^2)}}$ from $F(X)$ and $F(Y)$.
The sets $F(\alpha_{a,b}(X))$ and 
$F(\rho_{a \to b}(X))$  can be computed
in time $(n+1)^{2^{k^2}-1}\in n^{2^{\bigo(k^2)}}$ from $F(X)$.

In order to bound the number and order of operations within directed clique-width expressions,
we can use the normal form for clique-width expressions defined in  \cite{EGW03}.
The proof of Theorem 4.2 in \cite{EGW03} shows that also for directed  clique-width expression $X$,
we can assume that for every subexpression, after a disjoint union operation
first there is a sequence of edge insertion operations followed by a sequence of
relabeling operations, i.e. between two disjoint union operations there is no relabeling before
an edge insertion. Since there are $n$ leaves in $T$, we have $n-1$
disjoint union operations, at most $(n-1)\cdot (k-1)$ relabeling operations,
and at most $(n-1)\cdot k(k-1)$ edge insertion  operations.
This leads to an overall running time of $n^{2^{\bigo(k^2)}}$.
\end{proof}

\begin{example}
We consider $X=(\alpha_{1,2}(\rho_{2\to 1}(\alpha_{2,1}(\alpha_{1,2}(1(v_1)\oplus 2(v_2))))\oplus 2(v_3))$.
\begin{itemize}
\item
$F(1(v_1))=\{ \langle \{(1,1)\} \rangle \}$

\item
$F(2(v_2))=\{ \langle \{(2,2)\} \rangle \}$

\item
$F(1(v_1)\oplus 2(v_2))=\{ \langle \{(1,1)\},\{(2,2)\} \rangle, \langle \{(1,1),(2,2)\} \rangle \}$

\item
$F(\alpha_{1,2}(1(v_1)\oplus 2(v_2)))=\{ \langle \{(1,1)\},\{(2,2)\} \rangle, \langle \{(1,1),(2,2),(1,2)\} \rangle \}$

\item
$F(\alpha_{2,1}(\alpha_{1,2}(1(v_1)\oplus 2(v_2))))=\{ \langle \{(1,1)\},\{(2,2)\} \rangle\}$

\item
$F(\rho_{2\to 1}(\alpha_{2,1}(\alpha_{1,2}(1(v_1)\oplus 2(v_2)))))=\{ \langle \{(1,1)\},\{(1,1)\} \rangle\}$

\item
$F(\rho_{2\to 1}(\alpha_{2,1}(\alpha_{1,2}(1(v_1)\oplus 2(v_2))))\oplus 2(v_3))=\{ \langle \{(1,1)\},\{(2,2)\},\{(1,1)\} \rangle ,\langle \{(1,1),(2,2)\},$ $\{(1,1)\} \rangle \}$

\item
$F(\alpha_{1,2}(\rho_{2\to 1}(\alpha_{2,1}(\alpha_{1,2}(1(v_1)\oplus 2(v_2))))\oplus 2(v_3)))=\{ \langle \{(1,1)\},\{(2,2)\},\{(1,1)\} \rangle,$ $\langle \{(1,1),(2,2),(1,2)\},$ $\{(1,1)\} \rangle \}$

\end{itemize}

Thus, in $F(X)$ there is one multi set consisting of two  sets of label pairs 
and one multi set consisting of three  sets of label pairs. This implies that
$\vec{\chi}(\g(X))=2$.
\end{example}

The running time shown in Theorem \ref{xp-dca} leads to the following result.

\begin{corollary}\label{xp-dc}
The  Dichromatic Number problem  is in $\xp$ when parameterized by directed clique-width.
\end{corollary}

Up to now there are only very few digraph classes for which we can compute a directed clique-width expression in
polynomial time. This holds for directed co-graphs, digraphs of bounded directed modular width, 
orientations of trees, and directed cactus forests.
For such classes we can apply the result of Theorem~\ref{xp-dca}.
In order to find directed clique-width expressions for general
digraphs one can use results on the related parameter bi-rank-width \cite{KR13}.
By \cite[Lemma 9.9.12]{BG18} we can use approximate directed clique-width
expressions obtained from rank-decomposition with the
drawback of a  single-exponential blow-up on the parameter.

Next, we give a lower bound for the running time of parameterized
algorithms for Dichromatic Number  problem  parameterized by
the directed clique-width.

\begin{corollary}
The Dichromatic Number  problem on digraphs on $n$ vertices parameterized by
the directed clique-width $k$ cannot be solved in time $n^{2^{o(k)}}$, unless ETH fails.
\end{corollary}

\begin{proof}
In order to show the statement we apply the following lower bound for the
Chromatic Number problem parameterized by clique-width given in \cite{FGLSZ18}.
Any algorithm for the Chromatic Number problem parameterized by clique-width
with running in $n^{2^{o(k)}}$  would disprove the Exponential Time Hypothesis.
By  Observation \ref{obs-dicol} and
since for every undirected graph $G$ its clique-width equals the directed
clique-width of $\overleftrightarrow{G}$ \cite{GWY16}, any algorithm
for the  Dichromatic Number  problem  parameterized by
the directed clique-width can be used to solve the
Chromatic Number problem parameterized by clique-width.
\end{proof}

In order to show  fixed parameter tractability for  $\DCN_{r}$ w.r.t.\ the
parameter directed clique-width one can use its defineability within monadic second order logic (MSO).
We restrict to $\MSOA$-logic, which allows propositional logic,
variables for vertices and vertex sets of digraphs, the predicate $\arc(u,v)$ for arcs of digraphs, and
quantifications over vertices and vertex sets \cite{CE12}.
For defining optimization problems we use the $\LMSOA$  framework given in \cite{CMR00}.

The following theorem is from \cite[Theorem 4.2]{GHKLOR14}.

\begin{theorem}[\cite{GHKLOR14}]\label{th-ghk}
For every integer $k$ and $\MSOA$  formula $\psi$, every $\psi$-$\LMSOA$  optimization problem
is fixed-parameter tractable on digraphs of clique-width $k$, with the parameters $k$ and  length of the formula $|\psi|$.
\end{theorem}

Next, we will apply this result to $\DCN$.

\begin{theorem}\label{fpt-cw-r}
The Dichromatic Number problem
is in $\fpt$ when parameterized by directed clique-width and $r$.
\end{theorem}

\begin{proof} Let $G=(V,E)$ be a digraph.
We can define $\DCN_{r}$ by an $\MSOA$ formula
$$\psi=\exists V_1,\ldots, V_r: \left( \text{Partition}(V,V_1,\ldots,V_r) \wedge \bigwedge_{1\leq i \leq r} \text{Acyclic}(V_i) \right)$$
with
$$
\begin{array}{lcl}\text{Partition}(V,V_1,\ldots,V_r)&=&\forall v\in V: ( \bigvee_{1\leq i\leq r}v\in V_i) \wedge \\
 && \nexists  v\in V: ( \bigvee_{i\neq j,~ 1\leq i,j \leq r}(v\in V_i\wedge v\in V_j))
\end{array}
$$
and
$$
\begin{array}{lcl}
\text{Acyclic}(V_i)&=&\forall V'\subseteq V_i, V'\neq \emptyset: \exists v\in V' (\odeg(v)=0 \vee \odeg(v)\geq 2)
\end{array} $$
For the correctness we note the following. For every induced cycle $V'$ in
$G$ it holds that for every vertex $v\in V'$ we have $\odeg(v)=1$ in $G$.
This does not hold for non-induced cycles. But since for every cycle $V''$ in $G$ 
there is a subset $V'\subseteq V''$, such that
$G[V']$ is a cycle,  we can verify by $\text{Acyclic}(V_i)$ whether $G[V_i]$ is acyclic.
Since it holds that $|\psi|\in \bigo(r)$, the statement follows by the result of Theorem \ref{th-ghk}. 
\end{proof}

\begin{corollary}\label{fpt-dc}
For every  integer $r$ the $r$-Dichromatic Number problem
is in $\fpt$ when parameterized by directed clique-width.
\end{corollary}

\section{Conclusions and outlook}

The presented methods allow us to compute the dichromatic number
on special classes of digraphs in polynomial time.

The shown parameterized solutions of Corollary \ref{xp-dc} and Theorem \ref{fpt-cw-r}
also hold for any parameter which is
larger or equal than directed clique-width, such as the parameter directed modular width  \cite{SW20} (which
even allows an $\fpt$-algorithm by \cite{SW19,SW20})
and directed linear clique-width \cite{GR19c}. Furthermore, restricted to
semicomplete digraphs the shown parameterized solutions also hold
for directed path-width~\cite[Lemma 2.14]{FP19}.

Further, the hardness result of Corollary \ref{hpcw} rules out
 $\fpt$-algorithms for the Dichromatic
Number problem parameterized by width parameters
 which can be bounded by directed clique-width. Among these are
the clique-width and rank-width of the underlying undirected
graph, which also have been considered in \cite{Gan09} on the Oriented Chromatic
Number problem.

\begin{corollary}\label{cor-fpt-un}
The   Dichromatic Number problem  is not in $\fpt$ when parameterized by clique-width of the underlying undirected graph, unless  $\p=\np$.
\end{corollary}

From a parameterized point of view width parameters are so-called structural parameters, which
are measuring the difficulty of decomposing a graph into a special tree-structure.
Beside these, the standard parameter, i.e.\ the threshold value given in the instance,
is well studied. Unfortunately, for the Dichromatic Number problem the standard parameter
is the number of necessary colors $r$ and
does even not allow an $\xp$-algorithm, since $\DCN_{2}$ is NP-complete  \cite{MSW19}.

\begin{corollary}\label{cor-xp-r}
The   Dichromatic Number problem  is not in $\xp$ when parameterized by $r$, unless  $\p=\np$.
\end{corollary}

A positive result can be obtained for parameter ''number of vertices'' $n$.
Since integer linear programming is fixed-parameter tractable for the
parameter ''number of variables''  \cite{Len83}
the existence of an integer program for $\DCN$ using $\bigo(n^2)$ variables
implies an $\fpt$-algorithm for parameter $n$.

\begin{remark}
To formulate $\DCN$ for some directed graph $G=(V,E)$ as an integer program,
we introduce a binary variable $y_j\in\{0,1\}$, $j\in \{1,\ldots,n\}$, such that
$y_j=1$ if and only if color $j$ is used. Further, we use $n^2$ variables
$x_{i,j}\in\{0,1\}$, $i,j\in \{1,\ldots,n\}$, such that
$x_{i,j}=1$ if and only if vertex $v_i$ receives color  $j$.

In order to ensure that every color class is acyclic,
we will use the well known characterization that a digraph is acyclic if
and only if it has a topological vertex ordering. A {\em topological vertex ordering}
of a directed graph  is a linear ordering of its vertices such that  for every 
edge $(v_i,v_{i'})$ vertex $v_i$ comes before vertex $v_{i'}$ in the ordering. 
The existence of a topological ordering will be verified by  further integer 
valued variables $t_i\in \{0,\ldots,n-1\}$ realizing the ordering number of vertex $v_i$, 
$i\in \{1,\ldots,n\}$.
In order to define only feasible orderings for every color $j\in \{1,\ldots,n\}$,
for edge $(v_i,v_{i'})\in A$ with $x_{i,j}=1$ and $x_{i',j}=1$ we verify
$t_{i'}\geq t_i+1$ in condition (\ref{01-p3-cn0}).


\begin{eqnarray}
\text{minimize}   \sum_{j=1}^{n}   y_j  \label{01-p1-cn0}
\end{eqnarray}
subject to
\begin{eqnarray}
\sum_{j=1}^{n} x_{i,j}  & =  &1 \text{ for every }    i \in \{1,\ldots,n\}  \label{01-p2-cn0} \\
                x_{i,j}    & \leq & y_j         \text{ for every } i,j \in  \{1,\ldots,n\} \label{01-p-cn0} \\
t_{i'}                     & \geq & t_i+1 -  n\cdot \left(1- (x_{i,j}\wedge  x_{i',j})\right)  \text{ for every } (v_i,v_{i'})\in E,  j \in \{1,\ldots,n\} \label{01-p3-cn0} \\
y_j   &\in &  \{0,1\} \text{ for every } j \in  \{1,\ldots,n\}  \label{01-p4-cn0} \\
t_i   &\in & \{0,\ldots,n-1\}  \text{ for every } i \in \{1,\ldots,n\}   \label{01-pc-cn0} \\
x_{i,j}   &\in &  \{0,1\} \text{ for every } i,j \in \{1,\ldots,n\}   \label{01-p5-cn0}
\end{eqnarray}


Equations in (\ref{01-p3-cn0})  are not in propositional logic, but can be
transformed for integer programming \cite{Gur14}.

\end{remark}

\begin{corollary}\label{fpt-n}
The Dichromatic Number problem
is in $\fpt$ when parameterized by the number of vertices $n$.
\end{corollary}

It remains to verify whether the running time of our $\xp$-algorithm for $\DCN$ can
be improved to  $n^{2^{\bigo(k)}}$, which is possible for the Chromatic Number
problem by \cite{EGW01a,KR03}.
Further, it remains open whether the
hardness of Corollary \ref{hpcw} also holds for special digraph classes and for
directed linear clique-width~\cite{GR19c}.
Furthermore, the existence of an $\fpt$-algorithm for $\DCN_r$ w.r.t.\ parameter
clique-width of the underlying undirected graph is open, see 
Table \ref{tab}.

\section*{Acknowledgements} \label{sec-a}

The work of the second and third author was supported
by the Deutsche
Forschungsgemeinschaft (DFG, German Research Foundation) -- 388221852


\bibliographystyle{alpha}

\newcommand{\etalchar}[1]{$^{#1}$}

\end{document}